\documentclass{jgcc} 


\usepackage{hyperref}
\theoremstyle{plain} 
\usepackage{latexsym}
\usepackage{amsmath}
\usepackage{amssymb}
\usepackage{amsthm}
\usepackage{hyperref}
\usepackage{cite}
\usepackage{graphicx}
\usepackage{color}
\usepackage{bm}
\usepackage{xspace}
\usepackage{doi, url}

\renewcommand{\paragraph}[1]{\vspace{6pt} \noindent \textbf{#1}\xspace}

\theoremstyle{plain}
\newtheorem{theorem}{Theorem}[section]

\newtheorem*{theorem*}{Theorem}

\newtheorem{lemma}[theorem]{Lemma}

\newtheorem{proposition}[theorem]{Proposition}
\newtheorem{claim}[theorem]{Claim}

\theoremstyle{definition}

\newtheorem{definition}[theorem]{Definition}
\newtheorem{example}[theorem]{Example}

\usepackage{arydshln}

\newcommand{\GL}{\mathrm{GL}}
\newcommand{\F}{\mathbb{F}}

\newcommand{\C}{\mathbb{C}}
\newcommand{\R}{\mathbb{R}}

\newcommand{\tr}[1]{#1^{\mathrm{t}}}

\newcommand{\rk}{\mathrm{rk}}

\newcommand{\sgn}{\mathrm{sgn}}

\newcommand{\poly}{\mathrm{poly}}

\newcommand{\M}{\mathrm{M}}

\renewcommand{\S}{\mathrm{S}}

\newcommand{\tuple}[1]{\mathbf{#1}}
\newcommand{\tens}[1]{\mathtt{#1}}
\newcommand{\spa}[1]{\mathcal{#1}}

\newcommand{\cA}{\spa{A}}
\newcommand{\cB}{\spa{B}}

\newcommand{\tA}{\tens{A}}
\newcommand{\tB}{\tens{B}}

\newcommand{\tG}{\tens{G}}

\newcommand{\vA}{\tuple{A}}
\newcommand{\vB}{\tuple{B}}
\newcommand{\vC}{\tuple{C}}
\newcommand{\vD}{\tuple{D}}

\newcommand{\aut}{\mathrm{Aut}}
\newcommand{\Aut}{\aut}
\newcommand{\Adj}{\mathrm{Adj}}
\newcommand{\Isom}{\mathrm{Iso}}

\newcommand{\algprobm}[1]{{\sc #1}\xspace}

\newcommand{\PI}{\algprobm{PI}}
\newcommand{\hPI}{\algprobm{HPI}}
\newcommand{\AI}{\algprobm{AI}}
\newcommand{\TFE}{\algprobm{TFE}}

\newcommand{\TI}{\algprobm{TI}}

\newcommand{\DeeTIlong}{\algprobm{$d$-Tensor Isomorphism}}

\renewcommand{\cc}[1]{\mathrm{#1}}

\newcommand{\gbinom}[3]{{\genfrac{[}{]}{0pt}{}{#1}{#2}}_{#3}}

\newcommand{\too}%
{\xrightarrow{\text{\raisebox{-3pt}{$\sim$}}\,}}

\newcommand{\arXiv}[2]{\href{https://arxiv.org/abs/#1}{arXiv:#1} #2}

\usepackage[T1]{fontenc}
\def\DJ{{\hbox{D\kern-.8em\raise.15ex\hbox{--}\kern.35em}}}

\pdfoutput=1
\usepackage{lastpage}
\jgccdoi{14}{1}{1}{9431}  
\jgccheading{}{\pageref{LastPage}}{}{}%
{May~6,~2022}{Aug.~10,~2022}{}    

\begin{document}

\title[Average-case algorithms for polynomial isomorphism]{
Average-case algorithms for testing isomorphism of polynomials, algebras, and 
multilinear forms
}
\titlecomment{{\lsuper*}A preliminary version of this paper appeared 
at the 38th International Symposium on Theoretical Aspects of Computer Science 
(STACS 2021) \cite{GQT21}.}

\author[J.A.~Grochow]{Joshua A. Grochow}	
\address{Departments of Computer Science and Mathematics, 
University of Colorado---Boulder, Boulder, CO 80309-0430, United 
States.}	
\email{jgrochow@colorado.edu}  
\thanks{J. A. G. was partially supported by NSF grants DMS-1750319 and 
CCF-2047756.}	

\author[Y.~Qiao]{Youming Qiao}	
\address{Centre for Quantum Software and Information,
      School of Computer Science, Faculty of Engineering and Information
      Technology, University of Technology Sydney, NSW 2150, Australia.}	
\email{Youming.Qiao@uts.edu.au}  
\thanks{Y. Q. was partially supported by the Australian Research Council 
DP200100950.}	

\author[G.~Tang]{Gang Tang}	
\address{Centre for Quantum Software and Information,
    School of Computer Science, Faculty of Engineering and Information
    Technology, University of Technology Sydney, NSW 2150, Australia.}	
\urladdr{Gang.Tang-1@student.uts.edu.au}  
\thanks{G. T. was partially supported by the Australian Research Council
DP160101652.}	





\begin{abstract}
  \noindent We study the problems of testing isomorphism of polynomials, algebras, 
  and 
  multilinear forms. Our 
  first main 
  results are average-case algorithms for these problems. For 
  example, we develop 
  an algorithm that takes two cubic forms $f, g\in \F_q[x_1, 
  \dots, x_n]$, and decides whether $f$ and $g$ are 
  isomorphic 
  in time $q^{O(n)}$ \emph{for most} $f$. This average-case setting has direct 
  practical implications, having been studied in multivariate cryptography since 
  the 
  1990s.
  Our second result concerns the complexity of 
  testing equivalence of alternating trilinear forms. This problem is of interest 
  in both mathematics and cryptography. We show that this problem is 
  polynomial-time equivalent to testing equivalence of symmetric trilinear forms, 
  by showing that they are both
  \algprobm{Tensor Isomorphism}-complete (Grochow \& Qiao, \emph{ITCS}, 2021), 
  therefore is 
  equivalent 
  to testing isomorphism of cubic forms over most fields.
\end{abstract}

\maketitle


\section{Introduction}\label{sec:intro}

In this paper, we study isomorphism testing problems for polynomials, algebras, 
and multilinear forms. Our first set of results is algorithmic, namely presenting 
average-case algorithms for these problems (Section~\ref{subsec:algo}). Our second 
result is 
complexity-theoretic, concerning the problems of testing equivalence of symmetric 
and alternating 
trilinear forms (Section~\ref{subsec:complexity}).

\subsection{Average-case algorithms for polynomial isomorphism and more}
\label{subsec:algo}

\paragraph{The polynomial isomorphism problem.} Let $\F$ be a field, and let 
$X=\{x_1, \dots, x_n\}$ be a set of variables. Let 
$\GL(n, \F)$ be the general linear group consisting of $n\times n$ invertible 
matrices over $\F$. A natural group action of $A=(a_{i,j})\in \GL(n, \F)$ on the 
polynomial ring $\F[X]$ sends $f(x_1, \dots, x_n)$ to $f\circ 
A:=f(\sum_{j=1}^na_{1,j}x_j, \dots, 
\sum_{j=1}^na_{n, j}x_j)$. The \emph{polynomial isomorphism problem} (\PI) asks, 
given $f, g\in \F[X]$, whether there exists $A\in\GL(n, \F)$ such that $f=g\circ 
A$. In the literature, this problem was also called the polynomial 
equivalence problem \cite{AS05}. 

An important subcase of \PI is when the input 
polynomials are required to be 
homogeneous of degree $d$. In this case, this problem is referred to as the 
homogeneous polynomial isomorphism 
problem, denoted as $d$-\hPI. Homogeneous degree-$3$ (resp. degree-$2$) 
polynomials are also known as cubic (resp. quadratic) forms. 

In this article, we assume that a polynomial is represented in algorithms 
by its list of coefficients of the monomials, though other representations like 
algebraic circuits are also possible in this context \cite{Kay11}. Furthermore, we 
shall mostly restrict our attention to the case when the polynomial degrees are 
constant.

\paragraph{Motivations to study polynomial isomorphism.} The polynomial 
isomorphism problem has been studied in both multivariate cryptography 
and computational complexity. In 1996, inspired by the celebrated zero-knowledge 
protocol for graph isomorphism \cite{GMW91}, Patarin proposed to use \PI as the 
security basis of authentication and signature protocols \cite{Pat96}. This lead 
to a series of works 
on practical algorithms for \PI; see \cite{Bou11,BFP15,IQ17} and  
references therein. In 1997, Grigoriev studied shift-equivalence of polynomials, 
and discussed equivalence of polynomials under more general groups \cite{Gri97}.
In the early 2000s, Agrawal, Kayal and 
Saxena studied \PI from 
the computational complexity perspective. They related \PI with graph isomorphism 
and algebra 
isomorphism \cite{AS05,AS06}, and studied some special instances of 
\PI \cite{Kay11} as well 
as several related algorithmic tasks \cite{Sax06}. 

Despite these works, little progress has been made on algorithms \emph{with 
rigorous analysis} for the \emph{general} \PI. More specifically, Kayal's 
algorithm \cite{Kay11} runs in randomized polynomial time, works for the degree 
$d\geq 4$, and doesn't require the field to be finite. However, it only works in 
the \emph{multilinear} setting, namely when $f$ and $g$ are isomorphic to a common 
multilinear polynomial $h$. 
The algorithms from 
multivariate cryptography \cite{Bou11} either are heuristic, or need unproven 
assumptions.
While these works contain several nice ideas and insights, and their 
implementations show practical improvements, they are nonetheless heuristic in 
nature, and rigorously analyzing them seems difficult. Indeed, if any of these 
algorithms had worst-case analysis matching their heuristic performance, it would 
lead to significant progress on the long-open Group Isomorphism problem (see, 
e.g., \cite{LQ17,GQ_arxiv}).

\paragraph{Our result on polynomial isomorphism.} Our first result is an 
average-case algorithm with rigorous analysis for \PI 
over a finite field $\F_q$. As far as we know, this is the first non-trivial 
algorithm with rigorous analysis for \PI over finite fields. (The natural 
brute-force algorithm, namely enumerating all invertible matrices, 
runs in time $q^{n^2}\cdot \poly(n, \log q)$.)
Furthermore, the 
average-case setting is quite natural, as it is 
precisely the one studied 
multivariate cryptography. We shall elaborate on this further after stating our 
result. 

To state the result, let us define what a random polynomial means in this setting. 
Since we represent 
polynomials by 
their lists of coefficients, a random polynomial of 
degree $d$ 
is 
naturally the one whose coefficients of the monomials of degree $\leq d$ are 
independently uniformly 
drawn from $\F_q$. We also consider the homogeneous 
setting where only monomials of degree $=d$ are of interest.
\begin{theorem}\label{thm:main}
Let $d\geq 3$ be a constant. Let $f, g\in \F_q[x_1, \dots, x_n]$ be 
(resp. homogeneous)
polynomials of degree $\leq d$ (resp. $=d$). There exists an $q^{O(n)}$-time 
algorithm that decides whether $f$ and $g$ are isomorphic, for all but at most 
$\frac{1}{q^{\Omega(n)}}$ 
fraction of $f$. 

Furthermore, if $f$ and $g$ are isomorphic, then this algorithm also computes an 
invertible matrix $A$ which sends $f$ to $g$. 
\end{theorem}

Let us briefly indicate the use of this average-case setting in multivariate 
cryptography. In the authentication scheme described in 
\cite{Pat96}, the public key consists of two polynomials $f, 
g\in\F_q[x_1, \dots, x_n]$, where $f$ is a random polynomial, and $g$ is obtained 
by applying a random invertible matrix to $f$. Then $f$ and $g$ form the public 
key, and any isomorphism from $f$ to $g$ can serve as the private key. Therefore, 
the algorithm in 
Theorem~\ref{thm:main} can be used to recover a private key for most $f$.

\paragraph{Adapting the algorithm strategy to more isomorphism problems.} In 
\cite{AS05,AS06}, the algebra isomorphism problem (\AI) was studied and shown to 
be (almost) polynomial-time equivalent to \PI. In \cite{GQ_arxiv}, many more 
problems are demonstrated to be polynomial-time equivalent to \PI, including the 
trilinear form equivalence problem (\TFE). In these reductions, due to the blow-up 
of the parameters, the $q^{O(n)}$-time algorithm in Theorem~\ref{thm:main} does 
not translate to moderately exponential-time, average-case algorithms for these 
problems. The 
algorithm design \emph{idea}, however, does translate to give moderately 
exponential-time, 
average-case algorithms for \AI and \TFE. These will be presented in 
Section~\ref{subsec:adjust}.

\subsection{Complexity of symmetric and alternating trilinear form 
equivalence}\label{subsec:complexity}

\paragraph{From cubic forms to symmetric and alternating trilinear forms.} 
In the context of polynomial isomorphism, cubic forms are of 
particular interest. In complexity theory, it was shown that $d$-\hPI reduces to 
cubic form isomorphism over fields with $d$th roots of unity \cite{AS05,AS06}. In 
multivariate cryptography, cubic form isomorphism also received special attention, 
since using higher degree forms results in less efficiency in the cryptographic 
protocols. 

Just as quadratic forms are closely related with symmetric bilinear forms, cubic 
forms are closely related with symmetric trilinear forms. Let $\F$ be a field of 
characteristic not $2$ or $3$, and let $f=\sum_{1\leq i\leq j\leq k\leq 
n}a_{i,j,k}x_ix_jx_k\in\F[x_1, \dots, x_n]$ be a cubic form. For any $i, j, 
k\in[n]$, let $1\leq i'\leq j'\leq k'\leq n$ be the result of sorting $i, j, k$ in 
the increasing order, and set $a_{i,j,k}=a_{i',j',k'}$.
Then we can define a symmetric\footnote{That is, for any permutation 
$\sigma\in\S_3$, $\phi(u_1, u_2, u_3)=\phi(u_{\sigma(1)}, 
u_{\sigma(2)}, 
u_{\sigma(3)})$} trilinear form 
$\phi_f:\F^n\times\F^n\times\F^n\to\F$ 
by 
$$\phi_f(u, v, 
w)=\sum_{i\in[n]}a_{i,i,i}u_iv_iw_i+\frac{1}{3}\cdot \sum_{\substack{i,j,k\in[n] 
\\ \text{ two of } i,j,k\text{ are the 
same}}}a_{i,j,k}u_iv_jw_k+\frac{1}{6}\cdot \sum_{\substack{i, 
j, k\in[n] \\ i, j, k \text { all 
different}}}a_{i,j,k}u_iv_jw_k.$$
It can be
seen easily that for any $v=\tr{(v_1, \dots, v_n)}\in \F^n$, $f(v_1, \dots, 
v_n)=\phi_f(v, v, v)$. 

In the theory of bilinear forms, symmetric and skew-symmetric bilinear forms are 
two important special subclasses. For example, they are critical in the 
classifications of classical groups \cite{Wey97} and finite simple groups 
\cite{Wilson_book}. 
For trilinear forms, we also have skew-symmetric trilinear forms. In fact, to 
avoid some complications over fields of characteristics $2$ or $3$, we shall 
consider 
alternating trilinear forms which are closely related to skew-symmetric ones. 
For trilinear forms, the exceptional groups of type $E_6$ can be constructed as 
the stabilizer of certain symmetric trilinear forms, and those of type $G_2$ can 
be constructed as the stabilizer of certain alternating trilinear forms.

We say that 
a trilinear form $\phi:\F^n\times\F^n\times\F^n\to\F$ is \emph{alternating}, if 
whenever two 
arguments of 
$\phi$ are equal, $\phi$ 
evaluates to zero. Note that this implies skew-symmetry, namely for any $u_1, u_2, 
u_3\in \F^n$ and any 
$\sigma\in\S_3$, $\phi(u_1, 
u_2, u_3)=\sgn(\sigma)\cdot \phi(u_{\sigma(1)}, u_{\sigma(2)}, u_{\sigma(3)})$. 
Over fields of characteristic zero or $> 3$, this is equivalent to skew-symmetry.

\paragraph{The trilinear form equivalence problem.}
Given a trilinear 
form $\phi:\F^n\times\F^n\times\F^n\to\F$, $A\in \GL(n, \F)$ 
naturally acts on $\phi$ by sending it to $\phi\circ A:=\phi(A^{-1}(u), A^{-1}(v), 
A^{-1}(w))$. The \emph{trilinear form equivalence 
problem} then asks, given two trilinear forms $\phi, 
\psi:\F^n\times\F^n\times\F^n\to\F$, whether there exists $A\in\GL(n, \F)$, such 
that $\phi=\psi\circ A$. Over fields of 
characteristic not 
$2$ or $3$, two cubic forms $f$ and $g$ are isomorphic if and only if $\phi_f$ and 
$\phi_g$ are equivalent, so cubic form isomorphism is polynomial-time equivalent 
to symmetric trilinear form equivalence over such fields. Note that for clarity, 
we reserve the term
``isomorphism'' for 
polynomials (and 
cubic forms), and use ``equivalence'' for multilinear forms.

\paragraph{Motivations to study alternating trilinear form equivalence.} 
Our main interest is to study the 
complexity of alternating trilinear form equivalence, with the 
following motivations. 

The first motivation comes from cryptography. To store a symmetric trilinear form 
on $\F_q^n$, $\binom{n+2}{3}$ field elements are required. To store an alternating 
trilinear form on $\F_q^n$, $\binom{n}{3}$ field elements are needed. The 
difference between $\binom{n+2}{3}$ and $\binom{n}{3}$ could be significant for 
practical 
purposes. For example, when $n=9$, $\binom{n+2}{3}=\binom{11}{3}=165$, while 
$\binom{n}{3}=\binom{9}{3}=84$. This means that in the authentication protocol of 
Patarin \cite{Pat96}, using alternating trilinear forms instead of cubic forms for 
$n=9$,\footnote{The parameters of the cryptosystem are $q$ 
and $n$. When $q=2$, $n=9$ is not secure as it can be solved in practice 
\cite{BFFP11}. So $q$ needs to be large for $n=9$ to be secure. Interestingly, 
according to \cite[pp. 227]{Bou11}, the parameters $q=16$ and $n=8$ seemed 
difficult for practical attacks via Gr\"obner basis. }
one saves almost one half in the public key size, which is an important saving in 
practice. 

The second 
motivation originates from comparing symmetric and alternating bilinear forms. It 
is well-known that, \emph{in the bilinear case},
the structure of alternating forms is simpler than that for symmetric ones 
\cite{Lang}. 
Indeed, 
up to equivalence, an alternating 
bilinear 
form is completely 
determined by its rank over any 
field, while the classification of symmetric bilinear forms depends crucially on 
the underlying 
field. For example, recall that over $\R$, a symmetric 
form is determined by its ``signature'', so just the rank is not enough.  

A third motivation is implied by the representation theory of the general 
linear groups; namely that alternating trilinear forms are the ``last'' 
natural case for $d=3$. If we consider the action of $\GL(n,\C)$ acting on 
$d$-tensors in $\C^n \otimes \C^n \otimes \dotsb \otimes \C^n$ diagonally (that 
is, the same matrix acts on each tensor factor), it is a classical result 
\cite{Wey97} that the 
invariant subspaces of $(\C^n)^{\otimes d}$ under this action are 
completely determined by the irreducible representations of $\GL(n,\C)$. When 
$d=3$, there are only three such representations, which correspond precisely to: 
symmetric trilinear forms, Lie algebras, and alternating trilinear forms. From the 
complexity point of view, it was previously shown that isomorphism of symmetric 
trilinear forms \cite{AS05, AS06} and Lie algebras \cite{GQ_arxiv} are equivalent 
to algebra isomorphism. Here we show that the last case, isomorphism of 
alternating trilinear forms, is also equivalent to the others. 

\paragraph{The complexity of alternating trilinear form equivalence.} Given 
the above discussion on the comparison between symmetric and alternating bilinear 
forms, one may wonder whether alternating trilinear form 
equivalence was easier than symmetric trilinear form equivalence. Interestingly, 
we show that this is not the case; rather, 
 they are polynomial-time equivalent.

\begin{theorem}\label{thm:complexity}
The alternating trilinear form equivalence problem is polynomial-time equivalent 
to the symmetric trilinear form equivalence problem. 
\end{theorem}

Note here that the reduction from alternating to symmetric trilinear form 
equivalence requires us to go through the tensor isomorphism problem, which causes 
polynomial blow-ups in the dimensions of the underlying vector spaces. Therefore, 
though these two problems are polynomial-time equivalent, these problems may 
result in cryptosystems with different efficiencies for a given security level. 

\subsection{Previous works}

\paragraph{The relation between \PI and \AI.} As mentioned in 
Section~\ref{subsec:algo}, the degree-$d$ homogeneous polynomial 
isomorphism 
problem 
($d$-\hPI) was shown to be almost equivalent to the algebra isomorphism problem 
(\AI) 
in \cite{AS05,AS06}. (See Section~\ref{subsec:adjust} for the formal definition of 
algebra isomorphism problem.) Here, almost refers to that for the reduction from 
$d$-\hPI 
to 
\AI in \cite{AS05,AS06}, the underlying fields are required to contain a $d$th 
root of unity. When $d=3$, this means that the characteristic of the underlying 
field $p$ satisfies that $p=2 \mod 3$ or $p=0$, which amounts to half of the 
primes. 
In \cite{GQ_arxiv}, another reduction from $3$-\hPI to \AI was 
presented, which works for fields of characteristics not $2$ or $3$. The reduction 
from \AI to $3$-\hPI in \cite{AS06} works over any field. 

\paragraph{The tensor isomorphism complete class.}
In \cite{FGS19,GQ_arxiv}, polynomial-time equivalences are proved between 
isomorphism 
testing of many more 
mathematical structures, including tensors, matrix spaces, polynomial maps, and so 
on. These problems arise from many areas: besides multivariate cryptography and 
computational complexity, they appear in quantum information, machine learning, 
and computational group theory. This motivates the authors of \cite{GQ_arxiv} to 
define the tensor isomorphism complete class \TI, which we recall here.
\begin{definition}[{The \DeeTIlong problem, and the complexity class 
\TI}]\label{def:TI}
\DeeTIlong over a field $\F$ is the problem: 
given two $d$-way 
arrays $\tA = (a_{i_1,\dotsc,i_d})$ and $\tB=(b_{i_1, \dots, i_d})$, where 
$i_k\in[n_k]$ for $k\in[d]$, 
and $a_{i_1, \dots, i_d}, b_{i_1, 
\dots, i_d}\in \F$,
decide whether there are $P_k\in\GL(n_k,\F)$ for $k\in[d]$, such that for all 
$i_1,\dotsc,i_d$, 
\begin{equation}\label{eq:d_array_action}
a_{i_1,\dotsc,i_d} = \sum_{j_1,\dotsc,j_d} b_{j_1, \dotsc, j_d} (P_1)_{i_1, j_1} 
(P_2)_{i_2, j_2} \dotsb (P_d)_{i_d,j_d}.
\end{equation}

For any field $\F$, $\cc{TI}_\F$ denotes the class of problems that are 
polynomial-time Turing (Cook) reducible to \DeeTIlong over 
$\F$, for some $d$.
A problem is $\cc{TI}_\F$-complete, if it is in $\cc{TI}_\F$, and 
\DeeTIlong over $\F$ for any $d$ reduces to this problem. 

When a problem is naturally defined and is $\cc{TI}_\F$-complete over any $\F$, 
then we can simply write that it is $\cc{TI}$-complete.
\end{definition} 

The authors of \cite{GQ_arxiv} further utilised this connection between tensors 
and groups to show search-to-decision, counting-to-decision, and nilpotency class 
results for $p$-group isomorphism \cite{GQ21_CCC}.

\paragraph{Average-case algorithms for matrix space isometry.} In 
\cite{LQ17,BLQW20}, motivated by testing isomorphism of $p$-groups 
(widely believed to be the hardest cases of Group Isomorphism, see e.g. 
\cite{GQ17}), the 
algorithmic problem alternating matrix space isometry was studied. (In the 
literature \cite{Wil09a}, this problem was also known as the 
alternating bilinear map pseudo-isometry problem.) That problem asks the 
following: given two linear spaces of alternating matrices $\cA, \cB\leq\Lambda(n, 
q)$, decide whether there exists $T\in\GL(n, q)$, such that $\cA=\tr{T}\cB 
T=\{\tr{T}BT : B\in \cB\}$. (See Section~\ref{sec:prel} for the definition of 
alternating matrices.) The main result of \cite{BLQW20}, improving upon the 
one in \cite{LQ17}, is an 
average-case algorithm for this problem in time $q^{O(n+m)}$, where $m=\dim(\cA)$. 

\subsection{Remarks on the technical side}\label{subsec:tech}

\paragraph{Techniques for proving Theorem~\ref{thm:main}.} The algorithm 
for \PI in Theorem~\ref{thm:main} is based on the algorithmic idea from 
\cite{LQ17,BLQW20}. However, to adapt that idea to the \PI setting, there are
several interesting conceptual and technical difficulties. 

One conceptual 
difficulty is that for alternating matrix space isometry, there are actually two 
$\GL$ actions, one is by $\GL(n, q)$ as explicitly described above, and the other 
is by $\GL(m, q)$ performing the basis change of matrix spaces. The algorithm 
in \cite{BLQW20} crucially uses that the $\GL(m, q)$ action is ``independent'' of 
the $\GL(n, q)$ action. For \PI, there is only one $\GL(n, q)$-action acting 
on all the variables. Fortunately, as shown in Section~\ref{subsec:cubic}, there 
is 
still a natural way of applying the the 
basic idea from \cite{LQ17,BLQW20}.

One technical difficulty is that the analysis in \cite{BLQW20} relies on 
properties of random alternating matrices, while for $3$-\hPI, the analysis relies 
on properties of random symmetric matrices. To adapt the proof strategy in 
\cite{BLQW20} (based on \cite{LQ17}) to the symmetric setting is not difficult, 
but suggests some interesting differences between symmetric and alternating 
matrices (see the discussion after Claim~\ref{claim:LQ17}).

\paragraph{Techniques for proving Theorem~\ref{thm:complexity}.} By \cite{FGS19}, 
the trilinear form equivalence problem is in \TI, and so are the 
special cases symmetric and alternating trilinear form equivalence. The proof of 
Theorem~\ref{thm:complexity} goes by showing that both symmetric and alternating 
trilinear form equivalence are \TI-hard. 

Technically, the basic proof strategy is to adapt a 
gadget construction, which originates from \cite{FGS19} and then is further used 
in \cite{GQ_arxiv}. To use that gadget in the trilinear form setting does require 
several non-trivial ideas. First, we identify the right \TI-complete problem to 
start with, namely the alternating (resp. symmetric) matrix space isometry 
problem. Second, we need to arrange a $3$-way array $\tA$, representing a linear 
basis of an alternating (resp. symmetric) matrix spaces, into one representing an 
alternating trilinear form. This requires $3$ copies of $\tA$, assembled in an 
appropriate manner. Third, we need to add the gadget in three directions 
(instead of just two as in previous results). 
All 
these features were not present in \cite{FGS19,GQ_arxiv}. The 
correctness proof also requires certain tricky twists compared with those in 
\cite{FGS19} and 
\cite{GQ_arxiv}.

\paragraph{Structure of the paper.} In Section~\ref{sec:prel} we present certain 
preliminaries. In Section~\ref{sec:algo} we show average-case algorithms for 
polynomial isomorphism, algebra isomorphism, and trilinear form isomorphism, 
proving Theorem~\ref{thm:main}. In Section~\ref{sec:complexity} we prove 
Theorem~\ref{thm:complexity}.

\section{Preliminaries}\label{sec:prel}

\paragraph{Notations.} We collect the notations here, though some of them have 
appeared in Section~\ref{sec:intro}. Let $\F$ be a field. Vectors in $\F^n$ are 
column vectors. 
Let $e_i$ denote the $i$th standard basis vector of $\F^n$. Let 
$\M(\ell\times n, \F)$ be the linear 
space of $\ell\times n$ matrices over $\F$, and set $\M(n, \F):=\M(n\times n, 
\F)$. Let $I_n$ denote the identity matrix of size $n$. For $A\in \M(n, 
\F)$, $A$ is \emph{symmetric} if $\tr{A}=A$, and \emph{alternating} if for 
every $v\in \F^n$, $\tr{v}Av=0$. When the characteristic of $\F$ is not $2$, $A$ 
is 
alternating if and only if $A$ is skew-symmetric. Let $\S(n, \F)$ be the linear 
space of $n\times 
n$ symmetric matrices over $\F$, and let $\Lambda(n, \F)$ be the linear space of 
alternating matrices over $\F$. When $\F=\F_q$, 
we may write $\M(n, \F_q)$ as $\M(n, q)$. We use $\langle\cdot \rangle$ to denote 
the linear span.

\paragraph{$3$-way arrays.} A $3$-way array over a field $\F$ is an array with 
three indices whose elements are from $\F$. We use $\M(n_1\times n_2\times n_3, 
\F)$ to denote the linear space of $3$-way arrays of side lengths $n_1\times 
n_2\times n_3$ over $\F$. 

Let $\tA\in\M(\ell\times n\times m, \F)$. 
For $k\in[m]$, the $k$th 
\emph{frontal} slice of $\tA$ is $(a_{i,j,k})_{i\in[\ell], j\in[n]}\in 
\M(\ell\times n, 
\F)$. For $j\in[n]$, the $j$th \emph{vertical} slice of $\tA$ is 
$(a_{i,j,k})_{i\in[\ell], k\in[m]}\in 
\M(\ell\times m, \F)$. For $i\in[\ell]$, the $i$th \emph{horizontal} slice of 
$\tA$ is 
$(a_{i,j,k})_{j\in[n], k\in[m]}\in \M(n\times m, \F)$. We shall often think of 
$\tA$ as a matrix tuple in $\M(\ell\times n, \F)^m$ consisting of its frontal 
slices. 

A natural action of $(P, Q, R)\in \GL(\ell, \F)\times \GL(n, \F)\times \GL(m, \F)$ 
sends a 
$3$-way array $\tA\in \M(\ell\times n\times m, \F)$ to $\tr{P}\tA^RQ$, defined as 
follows. First represent $\tA$ as an $m$-tuple of $\ell\times n$ matrices 
$\vA=(A_1, \dots, A_m)\in\M(\ell\times n, \F)^m$. Then $P$ and $Q$ send $\vA$ to 
$\tr{P}\vA Q=(\tr{P}A_1Q, \dots, \tr{P}A_mQ)$, and $R=(r_{i,j})$ sends $\vA$ to 
$(A_1', \dots, A_m')$ where $A_i'=\sum_{j\in[m]}r_{i,j}A_j$. Clearly, the actions 
of $P$, $Q$, and $R$ commute. The resulting $m$-tuple of $\ell\times n$ matrices 
obtained by applying $P$, $Q$, and $R$ to $\vA$ is then $\tr{P}\vA^RQ$. Note that 
up to possibly relabelling indices, the entries of $\tr{P}\vA^RQ$ are explicitly 
defined as in Equation~\ref{eq:d_array_action}.

\paragraph{Useful results.} Let $\vA=(A_1, \dots, A_m), \vB=(B_1, \dots, 
B_m)\in\M(n, \F)^m$. Given $T\in \GL(n, \F)$, let $\tr{T}\vA T=(\tr{T}A_1 T, 
\dots, \tr{T}A_m T)$. We say that $\vA$ and $\vB$ are \emph{isometric}, if there 
exists 
$T\in\GL(n, \F)$ such that $\tr{T}\vA T=\vB$. Let $\Isom(\vA, \vB)=\{T\in\GL(n, 
\F) 
: \vA =\tr{T}\vB T\}$, and set $\Aut(\vA):=\Isom(\vA, \vA)$. Clearly, $\Aut(\vA)$ 
is a subgroup of $\GL(n, q)$, and $\Isom(\vA, \vB)$ is either empty or a coset of 
$\Aut(\vA)$.
 
\begin{theorem}[\cite{BW12,IQ17}]\label{thm:autometry-algorithms}
Let $\vA,\vB\in\S(n, q)^m$ (resp. $\Lambda(n,q)^m$) for some odd $q$. There exists 
a $\poly(n,m, 
q)$-time deterministic algorithm which takes $\vA$ and $\vB$ as inputs and outputs 
$\Isom(\vA,\vB)$, specified by (if nonempty) a generating set of $\Aut(\vA)$ (by 
the algorithm in~\cite{BW12}) and a coset representative $T\in\Isom(\vA,\vB)$ (by 
the algorithm in~\cite{IQ17}).
\end{theorem}

\section{Average-case algorithms for polynomial isomorphism and 
more}\label{sec:algo}

We shall present the algorithm for the cubic form isomorphism problem in detail in 
Section~\ref{subsec:cubic}. Based on this result, we present the proof of our main 
Theorem~\ref{thm:main} in 
Section~\ref{app:main}. The proof of the main technical lemma, 
Lemma~\ref{lem:tech}, is in Section~\ref{app:lem}. 
We will present our results for problems such as algebra 
isomorphism in 
Section~\ref{subsec:adjust}.

\subsection{Cubic form isomorphism over fields of odd order}\label{subsec:cubic}
We present the algorithm for cubic form isomorphism over 
fields of odd characteristic, as this algorithm already captures the essence of 
the idea, and cubic forms are most interesting from the \PI perspective as 
mentioned in Section~\ref{subsec:complexity}. A full proof of 
Theorem~\ref{thm:main}, which is a relatively minor extension of 
Theorem~\ref{thm:cubic_form}, is put in Section~\ref{app:main}. 

\begin{theorem}\label{thm:cubic_form}
Let $\F_q$ be a finite field of odd order, and $X=\{x_1, \dots, 
x_n\}$ be a set of commuting variables. Let $f, g\in \F_q[X]$ be 
two
cubic forms. There exists a deterministic algorithm that decides 
whether $f$ and $g$ are isomorphic in time $q^{O(n)}$, for all but at most 
$\frac{1}{q^{\Omega(n)}}$ fraction of $f$. 
\end{theorem}
\begin{proof}
Let $r$ be a constant to be determined later on, and suppose $n$ is sufficiently 
larger than $r$. Our goal is to find $T\in\GL(n, q)$, such that $f=g\circ T$. 

The algorithm consists of two main steps. Let us first give an overview of the two 
steps.

In the first step, we show that there exists a set of 
at most $q^{O(rn)}$-many $T_1\in\GL(n, 
q)$, such that every $T\in\GL(n, q)$ can be written as $T_1T_2$, 
where $T_2$ is of the form {\small \begin{equation}\label{eq:T_2}
\begin{bmatrix}
I_r & 0 \\
0 & R
\end{bmatrix}.
\end{equation} }
Furthermore, 
such $T_1$ can be enumerated in time $q^{O(rn)}$. We then set 
$g_1=g\circ T_1$. 

In the second step, we focus on searching for $T_2$ such that 
$f=g_1\circ T_2$. The key observation is that those $T_2$ as in 
Equation~\ref{eq:T_2} leave $x_i$, $i\in[r]$, invariant, and send $x_j$, 
$j\in[r+1, n]$, to a linear combination of $x_k$, $k\in[r+1, n]$. It follows that 
for any fixed $i\in[r]$, $T_2$ sends $\sum_{r+1\leq j\leq k\leq n} 
a_{i,j,k}x_ix_jx_k$ to a linear combination of $x_ix_jx_k$, 
$r+1\leq j\leq k\leq n$. 
We will use this observation to show that for a random $f$, the 
number of 
$T_2$ satisfying $f=g_1\circ T_2$ is upper bounded by $q^n$ with high probability. 
Furthermore, such 
$T_2$, if they exist, can be enumerated efficiently. This allows us to go over all 
possible $T_2$ and test if $f=g_1\circ T_2$.

\paragraph{The first step.} We show that there exist at most $q^{O(rn)}$-many 
$T_1\in\GL(n, q)$, such that any $T\in\GL(n, q)$ can be written as $T_1T_2$ where 
$T_2$ is of the form as in Equation~\ref{eq:T_2}.

Recall that $e_i$ is the $i$th standard basis vector. Let $E_r=\langle e_1, \dots, 
e_r\rangle$, and let $F_r=\langle e_{r+1}, \dots, e_n\rangle$. Suppose for 
$i\in[r]$, $T(e_i)=u_i$, and $T(F_r)=V\leq \F_q^n$. Let $T_1$ be any invertible 
matrix that 
satisfies $T_1(e_i)=u_i$, and $T_1(F_r)=V$. Let $T_2=T_1^{-1}T$. Then $T_2$ 
satisfies that for $i\in[r]$, $T_2(e_i)=e_i$, and $T_2(F_r)=F_r$. In other words, 
$T_2$ is of the form in Equation~\ref{eq:T_2}.

We then need to show that these $T_1$ can be enumerated in time $q^{O(rn)}$. 

Recall that $T_1$ is determined by the images of $e_i$, $i\in[r]$, and 
$F_r\leq\F_q^n$. So we first enumerate matrices of the form 
$\begin{bmatrix}
u_1 & \dots & u_r & e_{r+1} & \dots & e_n
\end{bmatrix}$, where $u_i\in \F_q^n$ are linearly independent. 
We then need to enumerate the possible images of $F_r$. Let 
$U=\langle u_1, \dots, u_r\rangle$. Then the image of $F_r$ is a complement 
subspace of $U$. It is well-known that the number of complement subspaces of a 
dimension-$r$ space 
is 
$\sim q^{r(n-r)}$. 
To enumerate all complement subspaces of $U$, first compute one complement 
subspace $V=\langle v_1, \dots, v_{n-r}\rangle$. Then it is easy to verify that, 
when going over $A=(a_{i,j})_{i\in[r], j\in[n-r]}\in\M(r\times (n-r), q)$, 
$\langle v_j+\sum_{i\in[r]}a_{i,j}u_i : j\in[n-r]\rangle$ go over all complement 
subspaces of $U$. It follows that we can enumerate matrices $T_1$ of the form 
$\begin{bmatrix}
u_1 & \dots & u_r & v_1+\sum_{i\in[r]}a_{i,1}u_i & \dots & 
v_{n-r}+\sum_{i\in[r]}a_{i,n-r}u_i
\end{bmatrix}$.

\paragraph{The second step.} In Step 1, we computed a set of invertible matrices 
$\{T_1\}\subseteq\GL(n, q)$ such that every $T\in\GL(n, q)$ can be written as 
$T=T_1T_2$ where $T_2=\begin{bmatrix}
I_r & 0 \\
0 & R
\end{bmatrix}$. So we set $h:=g\circ T_1$ and focus on finding $T_2$ of the 
above form such that $f=h\circ T_2$. 

Suppose $f=\sum_{1\leq i\leq j\leq k\leq n}\alpha_{i,j,k} x_ix_jx_k$, and 
$h=\sum_{1\leq i\leq j\leq k\leq n}\beta_{i,j,k} x_ix_jx_k$. 
For $i\in[r]$, define $f_i=\sum_{r+1\leq j\leq k\leq n}\alpha_{i,j,k}x_ix_jx_k$. 
Similarly define $h_{i}$. 

The key observation is that, due to the form of $T_2$, 
we have that $f_i=h_{i}\circ T_2$. This is because for $i\in[r]$, $T_2$ sends 
$x_i$ to $x_i$, and for $j\in[r+1, n]$, $T_2$ sends $x_j$ to a linear combination 
of $x_k$, $k\in[r+1, n]$. 

Let $\ell=n-r$. We then rename the variable $x_{r+i}$, $i\in[\ell]$ as $y_i$. Let 
$Y=\{y_1, \dots, 
y_{\ell}\}$. Then from $f$, we define $r$ quadratic forms in $Y$, 
\begin{equation}\label{eq:c_i}
\forall i\in[r],  c_i=\sum_{1\leq j\leq 
k\leq 
\ell}\alpha_{i,j,k}'y_jy_k, \text{ where }
\alpha_{i,j,k}'=\alpha_{i,r+j, r+k}. 
\end{equation}
Correspondingly, we define $r$ 
quadratic forms $d_i=\sum_{1\leq j\leq k\leq \ell}\beta_{i,j,k}'y_jy_k$, 
$i\in[r]$, 
from $g_1$. 

Our task now is to search for the $R\in\GL(\ell, q)$ such that for every 
$i\in[r]$, $c_i=d_i\circ R$. 

To do that, we adopt the classical representation of quadratic forms as symmetric 
matrices. Here we use the assumption that $q$ is odd. Using the classical 
correspondence between quadratic forms and symmetric matrices, from $c_i$ we 
construct 
\begin{equation}\label{eq:C_i}
C_i=\begin{bmatrix}
\alpha_{i,1,1}' & \frac{1}{2} \alpha_{i,1,2}' & \dots & 
\frac{1}{2}\alpha_{i,1,\ell}'\\
\frac{1}{2}\alpha_{i,1,2}' & \alpha_{i,2,2}' & \dots & 
\frac{1}{2}\alpha_{i,2,\ell}'\\
\vdots & \vdots & \ddots & \vdots \\
\frac{1}{2}\alpha_{i,1,\ell}' & \frac{1}{2}\alpha_{i,2,\ell}' & \dots & 
\alpha_{i,\ell,\ell}'
\end{bmatrix}\in \S(\ell, q).
\end{equation}
Similarly define $D_i$ from $d_i$. It is classical that $c_i=d_i\circ 
R$ if and only if $C_i=\tr{R}D_iR$.

Let $\vC=(C_1, \dots, C_r)\in \S(\ell, q)^r$, and $\vD=(D_1, \dots, D_r)\in 
\S(\ell, q)^r$. 
Recall that $\Aut(\vC)=\{ R\in\GL(\ell, \F) : \tr{R}\vC R=\vC\}$, and 
$\Isom(\vC, \vD)=\{ R\in\GL(\ell, \F) : \vC =\tr{R}\vD R\}$. Clearly, $\Isom(\vC, 
\vD)$ 
is a (possibly empty) coset of $\Aut(\vC)$. When $\Isom(\vC, \vD)$ is 
non-empty, 
$|\Isom(\vC, \vD)|=|\Aut(\vC)|$. Our main technical lemma is the following, 
obtained 
by adapting certain results in \cite{LQ17,BLQW20} to the symmetric matrix setting.
Its proof is postponed to Section~\ref{app:lem}. 
\begin{lemma}\label{lem:tech}
Let $\vC=(C_1, \dots, C_8)\in\S(\ell, q)^8$ be a random symmetric matrix tuple. 
Then we have $|\Aut(\vC)|\leq q^\ell$ for all but at most 
$\frac{1}{q^{\Omega(\ell)}}$ fraction of such $\vC$.
\end{lemma}

Given this lemma, we can use Theorem~\ref{thm:autometry-algorithms} to decide 
whether $\vC$ and $\vD$ are isometric, and if so, compute $\Isom(\vC, \vD)$ 
represented as a coset in $\GL(\ell, q)$. By Lemma~\ref{lem:tech}, for all but at 
most $\frac{1}{q^{\Omega(\ell)}}$ fraction of $\vC$, $|\Isom(\vC, \vD)|\leq 
q^\ell\leq q^n$. 
With $\Isom(\vC, \vD)$ as a coset at hand, we can enumerate all 
elements in $\Aut(\vC)$ by the standard recursive closure algorithm \cite{Luks} 
and therefore all elements in $\Isom(\vC, \vD)$. We then either conclude that 
$|\Isom(\vC, \vD)|>q^n$, or have all $\Isom(\vC, \vD)$ at hand. In the former case 
we conclude that $\vC$ does not satisfy the required generic condition. In the 
latter case, we enumerate $R\in \Isom(\vC, \vD)$, and check whether 
$T_2=\begin{bmatrix}
I_r & 0 \\
0 & R
\end{bmatrix}$ is an isomorphism from $f$ to $g_1$.

\paragraph{The algorithm outline.} We now summarise the above steps in the 
following algorithm outline. In the following we assume that $n\gg 8$; otherwise 
we 
can use the brute-force algorithm.
\begin{description}
\item[Input] Cubic forms $f, g\in \F_q[x_1, \dots, 
x_n]$. 
\item[Output] One of the following: (1) ``$f$ does not satisfy the generic 
condition''; (2) ``$f$ and 
$g$ are not isomorphic''; (3) an isomorphism $T\in\GL(n, q)$ sending $g$ to $f$.
\item[Algorithm outline] 
\begin{enumerate}
\item Set $r=8$, and $\ell= n-r$.
\item\label{enum:W} Compute $W=\{T_1\}\subseteq \GL(n, q)$ using the procedure 
described in Step 
1. \\
// {\tt Every $T\in\GL(n, q)$ can be written as $T_1T_2$ where $T_2$ is of the 
form in Equation~\ref{eq:T_2}.}
\item For every $T_1\in W$, do the following:
\begin{enumerate}
\item $h:\gets g\circ T_1$. 
\item For $i\in[\ell]$, $y_i\gets x_{r+i}$.
\item For $i\in[r]$, let $C_i\in \S(\ell, q)$ be defined in Equation~\ref{eq:C_i}. 
Let $D_i\in \S(\ell, q)$ be defined from $h$ in the same way. Let $\vC=(C_1, 
\dots, C_r)$, and $\vD=(D_1, \dots, D_r)$. 
\item\label{enum:R} Use Theorem~\ref{thm:autometry-algorithms} to decide whether 
$\vC$ and $\vD$ 
are isometric. If not, break from the loop. If so, compute one isometry $R$.
\item\label{enum:Aut} Use Theorem~\ref{thm:autometry-algorithms} to 
compute a generating set of 
$\Aut(\vC)$. Use the recursive closure algorithm to enumerate $\Aut(\vC)$. During 
the enumertion, if $|\Aut(\vC)|>q^\ell$, report ``$f$ does not satisfy the generic 
condition.'' Otherwise, we have the whole $\Aut(\vC)$ at hand, which is of size 
$\leq q^\ell$. 
\item\label{enum:verify} Given $R$ from Line~\ref{enum:R} and $\Aut(\vC)$ from 
Line~\ref{enum:Aut}, 
the whole set $\Isom(\vC, \vD)$ can be computed. For every $R\in \Isom(\vC, \vD)$, 
check whether $T_2=\begin{bmatrix}
I_r & 0 \\
0 & R
\end{bmatrix}$ sends $h$ to $f$. If so, return $T=T_1T_2$ as an isomorphism 
sending $g$ to $f$. 
\end{enumerate}
\item Return that ``$f$ and 
$g$ are not isomorphic''.
\end{enumerate}
\end{description}

\paragraph{Correctness and timing analyses.} The correctness of the algorithm 
relies on the simple fact that if $f$ satisfies the genericity condition, and $f$ 
and $g$ are isomorphic via some $T\in\GL(n, q)$, then this $T$ can be decomposed 
as $T_1T_2$ for some $T_1\in W$ from Line~\ref{enum:W}. Then by the analysis in 
Step 2, $T_2=\begin{bmatrix}
I_r & 0 \\
0 & R
\end{bmatrix}$ where $R\in\Isom(\vC, \vD)$. When $f$ satisfies the genericity 
condition, $\Isom(\vC, \vD)$ will be enumerated, so this $R$ will surely be 
encountered. 

To estimate the time complexity of the algorithm, note that $|W|\leq q^{O(rn)}$, 
and $|\Isom(\vC, \vD)|\leq q^\ell=q^{n-r}$. As other steps are performed in time 
$\poly(n, m, q)$, enumerating over $W$ and $\Isom(\vC, \vD)$ dominates the 
time complexity. Recall that $r=8$. So the total time complexity is upper bounded 
by $q^{O(n)}$.
\end{proof}

\subsection{Proof of the remaining cases of Theorem~\ref{thm:main}}\label{app:main}

Given Theorem~\ref{thm:cubic_form}, we can complete the proof of 
Theorem~\ref{thm:main} easily.
\begin{proof}
\paragraph{Cubic forms over fields of characteristic $2$.}
In Theorem~\ref{thm:cubic_form} we solved the case for cubic forms over fields 
of odd orders. 
We now consider cubic forms over fields of characteristic $2$. 

In this case, one difficulty is that we cannot use the correspondence between 
quadratic forms and symmetric matrices as used in Equation~\ref{eq:C_i}. Still, 
this difficulty can be overcome as follows. Let $f=\sum_{1\leq i\leq j\leq k\leq 
n}\alpha_{i,j,k}x_ix_jx_k$ where $\alpha_{i,j,k}\in\F_q$, $q$ is a power of $2$. 
We follow the proof strategy of Theorem~\ref{thm:cubic_form}. Step 1 stays 
exactly the same. In Step 2, we have $f$ and $g_1$, and the question is to look 
for $T_2=\begin{bmatrix}
I_r & 0 \\
0 & R
\end{bmatrix}$ such that $f=g_1\circ T_2$. We still consider the quadratic forms 
$c_i=\sum_{1\leq j\leq 
k\leq \ell}\alpha_{i,j,k}y_jy_k$ for $i\in[r]$. Now note that 
$(\sum_{j\in[\ell]}\beta_jy_j)^2=\sum_{j\in[\ell]}\beta_j^2y_j^2$ over fields of 
characteristic $2$. So the monomials $y_j^2$ do not contribute to $y_jy_k$ for 
$j\neq k$ under linear transformations. It follows that we can restrict our 
attention to $c_i'=\sum_{1\leq j< k\leq \ell}\alpha_{i,j,k}'y_jy_k$ for $i\in[r]$, 
and define alternating matrices \begin{equation}\label{eq:C_i_prime}
C_i=\begin{bmatrix}
0 & \alpha_{i,1,2}' & \dots & \alpha_{i,1,\ell}'\\
\alpha_{i,1,2}' & 0 & \dots & \alpha_{i,2,\ell}'\\
\vdots & \vdots & \ddots & \vdots \\
\alpha_{i,1,\ell}' & \alpha_{i,2,\ell}' & \dots & 0
\end{bmatrix}
\end{equation} for $i\in [r]$ to get $\vC\in\Lambda(\ell, q)^r$. Note that $C_i$ 
is alternating because we work over fields of characteristic $2$. Similarly 
construct 
$\vD\in \Lambda(\ell, q)^r$ from $g_1$. It can then be verified that, for 
$T_2=\begin{bmatrix}
I_r & 0 \\
0 & R
\end{bmatrix}$ to be an isomorphism from $g_1$ to $f$, it is necessary that $R$ is 
an isometry from $\vD$ to $\vC$. We then use \cite[Proposition 12]{BLQW20}, which 
is the alternating matrix version of our Lemma~\ref{lem:tech}. That proposition 
ensures that for $r=20$, all but at most $\frac{1}{q^{\Omega(\ell)}}$ fraction of 
$\vC$ has $|\Aut(\vC)|\leq q^\ell$. This explains how the first difficulty is 
overcome. 

However, there is a second difficulty, namely 
Theorem~\ref{thm:autometry-algorithms} do not apply to fields of characteristic 
$2$. We sketch how to overcome this difficulty here. The key is to look into the 
proof of 
\cite[Proposition 12]{BLQW20}, which in fact ensures that $\Adj(\vC)=\{(A, E)\in 
\M(\ell, q)\oplus \M(\ell, q) \mid \tr{A}\vC=\vC E\}$ is of size $\leq q^\ell$ for 
random $\vC$. Note that $\Adj(\vC)$ is a linear space and a linear basis of 
$\Adj(\vC)$ can be solved efficiently. Therefore, replacing $\Aut(\vC)$ with 
$\Adj(\vC)$ and $\Isom(\vC, \vD)$ with $\Adj(\vC, \vD)=\{(A, E)\in \M(\ell, 
q)\oplus \M(\ell, q) \mid \tr{A}\vC=\vD E\}$, we can proceed as in the proof of 
Theorem~\ref{thm:cubic_form}. The interested readers may refer to 
\cite{BLQW20} for the details.

\paragraph{Degree-$d$ forms.} Let us then consider degree-$d$ forms. In this case, 
we 
follow the proof of Theorem~\ref{thm:cubic_form}. Step 1 stays exactly the 
same. In Step 2, instead of $\sum_{1\leq j\leq k\leq n}\alpha_{i,j,k}x_ix_jx_k$ 
for $i\in[r]$, we work with $\sum_{1\leq j\leq k\leq 
n}\alpha_{i,j,k}x_i^{d-2}x_jx_k$, noting that matrices in the form in 
Equation~\ref{eq:T_2} preserve the set of monomials $\{x_i^{d-2}x_jx_k\}$. Then 
for odd $q$ case, construct symmetric matrices as in Equation~\ref{eq:C_i} 
and proceed as in the rest of Theorem~\ref{thm:cubic_form}. For the even $q$ 
case, construct alternating matrices as in Equation~\ref{eq:C_i_prime}, and 
procees as described above. 

\paragraph{Degree-$d$ polynomials.} We now consider degree-$d$ polynomials. In 
this case, we can single out the degree-$d$ piece and work as in degree-$d$ form 
case. The only change is that in the verification step, we need to take into 
account the monomials of degree $<d$ as well. 

This concludes the proof of Theorem~\ref{thm:main}.
\end{proof}

\subsection{Proof of Lemma~\ref{lem:tech}}\label{app:lem}

Recall that $\vC=(C_1, \dots, C_8)\in \S(\ell, q)^8$ is a tuple of random 
symmetric matrices, and $\Aut(\vC)=\{R\in\GL(\ell, q) : \tr{R}\vC R=\vC\}$. Our 
goal is to prove that $|\Aut(\vC)|\leq q^\ell$ for all but 
at most $\frac{1}{q^{\Omega(\ell)}}$ fraction of random $\vC$. 

Let $\Adj(\vC):=\{(R, S)\in \M(\ell, q)\oplus \M(\ell, q) : \tr{R}\vC=\vC S\}$. It 
is clear that $|\Aut(\vC)|\leq|\Adj(\vC)|$. We will in fact prove that with high 
probability, $|\Adj(\vC)|\leq q^\ell$. The proof of the following mostly follows 
the proofs for general matrix spaces as in \cite{LQ17} and alternating matrix 
spaces as in \cite{BLQW20}.

To start with, we make use the following 
result from \cite{LQ17}. We say that $\vD=(D_1, \dots, D_r)\in\M(\ell, q)^r$ is 
\emph{stable}\footnote{
Note 
that the stable notion in \cite{LQ17} deals with a more general setting when the 
matrices are not necessarily square. Our definition here coincides with the one in 
\cite{LQ17} when restricting to square matrices. }, if for any $U\leq \F_q^\ell$, 
$1\leq \dim(U)\leq \ell-1$, 
$\dim(\vD(U))>\dim(U)$, where $\vD(U)=\langle \cup_{i\in[r]}D_i(U)\rangle$, and 
$D_i(U)$ denotes the image of $U$ under $D_i$. 
\begin{claim}[{\cite[Prop. 10 in the arXiv version]{LQ17}}]\label{claim:LQ17}
If $\vD\leq\M(\ell, q)$ is stable, then $|\Adj(\vD)|\leq q^\ell$.
\end{claim}

Therefore, we turn to show that a random $\vC\in\S(\ell, q)^8$ is stable with high 
probability. This was shown for random matrix tuples in $\M(\ell, q)^4$ in 
\cite{LQ17}, and random alternating matrix tuples in $\Lambda(\ell, q)^{16}$ in 
\cite{BLQW20}. The proof strategy for the symmetric case is similar, but certain 
differences between the symmetric and alternating matrices do arise, as reflected 
in the following.

Our goal is to show that 
$$
\Pr[\vC\in\S(\ell, q)^8\text{ is not stable}]\leq \frac{1}{q^{\Omega(\ell)}}.
$$
By definition, we have
$$
\Pr[\vC\in\S(\ell, q)^8\text{ is not stable}]=\Pr[\exists U\leq\F_q^\ell, 
1\leq\dim(U)\leq n-1, \dim(U)\geq\dim(\vC(U))].
$$
By union bound, we have 
\begin{eqnarray*}
& & \Pr[\exists U\leq\F_q^\ell, 1\leq\dim(U)\leq n-1, \dim(U)\geq\dim(\vC(U))]\\
&\leq  & \sum_{U\leq\F_q^\ell, 1\leq\dim(U)\leq n-1}\Pr[\dim(U)\geq\dim(\vC(U))].
\end{eqnarray*}

For $d\in[\ell-1]$, let $E_d=\langle e_1, \dots, e_d\rangle$. Let 
$U\leq\F_q^\ell$, 
$\dim(U)=d$. We claim that 
$\Pr[\dim(U)\geq\dim(\vC(U))]=\Pr[\dim(E_d)\geq\dim(\vC(E_d))]$. To see this, note 
that there exists $P\in\GL(\ell, q)$ such that $P(E_d)=U$. Then observe that 
$\dim((\tr{P}\vC P)(E_d))=\dim(\vC(U))$. It follows that $\dim(\vC(U))\leq\dim(U)$ 
if and only if $\dim((\tr{P}\vC P)(E_d))\leq\dim(E_d)$. The claim then follows, by 
observing that the map $\S(\ell, q)^r\to \S(\ell, q)^r$ via $\tr{P}\cdot P$ is 
bijective. As a consequence, for any $d\in[n-1]$, we have
$$
\sum_{U\leq\F_q^\ell, \dim(U)=d}\Pr[\dim(U)\geq\dim(\vC(U))]
=\gbinom{\ell}{d}{q}\cdot \Pr[\dim(\vC(E_d))\leq d].$$

Let $C_i^d$ be the submatrix of $C_i$ consisting of the first $d$ columns of 
$C_i$, and let $C^d=\begin{bmatrix}
C_1^d & \dots & C_r^d
\end{bmatrix}\in\M(\ell\times rd, q)$. Then $\dim(\vC(E_d))=\rk(C^d)$. 
Note that each $C_i^d$ is of the form 
\begin{equation}\label{eq:C_i^d}
\begin{bmatrix}
C_{i,1}^d\\
C_{i,2}^d
\end{bmatrix}
\end{equation}
where $C_{i,1}^d$ is a random symmetric matrix of size $d\times 
d$, and $C_{i,2}^d$ is a random matrix of size $(\ell-d)\times d$. 

We then need to prove the following result, from which our desired result would 
follow. Here we set $r=8$.
\begin{proposition}\label{prop:tech}
Let $C^d\in\M(\ell\times 8d, q)$ be in the form above. Then we have 
$\gbinom{\ell}{d}{q}\cdot\Pr[\rk(C^d)\leq d]\leq \frac{1}{q^{\Omega(\ell)}}$, 
\end{proposition}

To prove Proposition~\ref{prop:tech}, we utilise 
the following result from \cite{LQ17}.
\begin{proposition}[{\!\!\cite[Proposition 20]{LQ17}}]\label{prop:4_random}
Let $D\in\M(\ell\times 4d, q)$ be a random matrix, where $1\leq d\leq \ell-1$. 
Then 
$\gbinom{\ell}{d}{q}\cdot\Pr[\rk(D)\leq d]\leq \frac{1}{q^{\Omega(\ell)}}$.
\end{proposition}

To use the above result in our setting, however, there is a caveat caused by the 
symmetric structure of $C_{i, 1}^d$ for $i\in[r]$. This is resolved by observing 
the following claim, which basically says that we can simulate one random matrix 
in $\M(d, q)$ using two random symmetric matrices in $\S(d, q)$.
\begin{claim}\label{claim:merge}
Let $X$ and $Y$ be two random symmetric matrices from $\S(d, q)$, i.e. $$
	X=\left[
	\begin{matrix}
	x_{1,1}      & x_{1,2}       & \cdots & x_{1,d}      \\
	x_{1,2}      & x_{2,2}      & \cdots & x_{2,d}      \\
	\vdots & \vdots & \ddots & \vdots \\
	x_{1,d}     & x_{2,d}      & \cdots & x_{d,d}      \\
	\end{matrix}
	\right], 
	Y=\left[
	\begin{matrix}
	y_{1,1}      & y_{1,2}       & \cdots & y_{1,d}      \\
	y_{1,2}      & y_{2,2}     & \cdots & y_{2,d}      \\
	\vdots & \vdots & \ddots & \vdots \\
	y_{1,d}     & y_{2,d}      & \cdots & y_{d,d}     \\
	\end{matrix}
	\right]
	$$  
    Then $$
    Z=\left[
    	\begin{matrix}
    	x_{1,1}+y_{1,2}      & x_{1,2}+y_{1,3}       & \cdots & x_{1,d} 
    	+y_{1,1}     \\
    	x_{1,2}+y_{2,2}      &x_{2,2}+ y_{2,3}      & \cdots & x_{2,d} 
    	+y_{1,2}     \\
    	\vdots & \vdots & \ddots & \vdots \\
    	x_{1,d}+y_{2,d}     & x_{2,d}+y_{3,d}      & \cdots & x_{d,d}+y_{1,d}     
    	\\
    	\end{matrix}
    	\right]
    	$$  
        is a uniformly sampled random matrix in $\M(d, q)$, when $X$ and $Y$ are 
        sampled in uniformly random from $\S(d, q)$.
\end{claim}
\begin{proof}
Let $z_{i,j}$ be the $(i,j)$th entry of $Z$.
Note that each $x_{i,j}$ (resp. $y_{i,j}$), $i\neq j$, appear exactly twice in 
$Z$ on an antidiagonal $z_{1, i}, z_{2, i-1}, \dots, z_{i-1, 1}, z_{i, n}, z_{i+1, 
n-1}, \dots, z_{n, i+1}$. So we can focus on such an antidiagonal to show that 
when $x_{i,j}$ and $y_{i,j}$ are uniformly sampled from $\F_q$, $z_{i,j}$ are 
also uniformly sampled from $\F_q$. 

Let us first consider the case when $d$ is odd. Let us consider a specific one, 
say $z_{1,1}=x_{1,1}+y_{1,2}$, 
$z_{2,d}=x_{2,d}+y_{1,2}$, \dots, $z_{d, 2}=x_{2,d}+y_{3,d}$. Other antidiagonals 
are of the same structure. It can be verified 
that this is a system of $d$ linear equations in $d+1$ variables of rank $d$. It 
follows that when those $x_{i,j}$ and $y_{k,\ell}$ involved are sampled in uniform 
random from $\F_q$, $z_{i',j'}$ are also in uniformly random distribution.

The case when $d$ is even can be verified similarly. This concludes the proof.
\end{proof}

We are now ready to prove Proposition~\ref{prop:tech}.
\begin{proof}[{Proof of Proposition~\ref{prop:tech}}]
Recall that $C^d=\begin{bmatrix}
C_1^d & \dots & C_8^d
\end{bmatrix}$, where $C_i^d\in\M(\ell\times d, q)$ is of the form in 
Equation~\ref{eq:C_i^d}. For 
$i\in[4]$, let $C_i'^d\in\M(\ell\times d, q)$ be constructed from $C_{2i-1}^d, 
C_{2i}^d$ as in Claim~\ref{claim:merge}, and set $C'^d=\begin{bmatrix}
C_1'^d & \dots & C_4'^d
\end{bmatrix}$. It is clear that $\rk(C^d)\geq \rk(C'^d)$, so 
$\gbinom{\ell}{d}{q}\cdot\Pr[\rk(C^d)\leq 
d]\leq\gbinom{\ell}{d}{q}\cdot\Pr[\rk(C'^d)\leq d]$. By 
Claim~\ref{claim:merge}, $C'^d$ is a random matrix in $\M(\ell\times 4d, q)$. By 
Proposition~\ref{prop:4_random}, $\gbinom{\ell}{d}{q}\cdot\Pr[\rk(C'^d)\leq d]\leq 
\frac{1}{q^{\Omega(\ell)}}$. This concludes the proof.
\end{proof}

\subsection{Trilinear form equivalence and algebra isomorphism}
\label{subsec:adjust}

We present our results on trilinear form equivalence and algebra isomorphism, and 
only sketch the proofs because they mostly follow that for 
Theorem~\ref{thm:cubic_form}. 

\paragraph{Trilinear form equivalence.} 
The trilinear form equivalence problem was stated in 
Section~\ref{subsec:complexity}. In algorithms, a trilinear form $f$ is naturally 
represented as a $3$-way array $\tA=(a_{i,j,k})$ where $a_{i,j,k}=f(e_i,e_j,e_k)$. 
A random trilinear form over $\F_q$ 
denotes the setting when $\alpha_{i,j,k}$ are independently sampled from 
$\F_q$ uniformly at random.

\begin{theorem}\label{prop:trilinear}
Let $f:\F_q^n\times\F_q^n\times\F_q^n\to\F_q$ be 
a random trilinear form, and let $g:\F_q^n\times\F_q^n\times\F_q^n\to\F_q$ be an 
arbitrary trilinear form. There exists a deterministic algorithm that decides 
whether $f$ and $g$ are equivalent in time $q^{O(n)}$, for all but at most 
$\frac{1}{q^{\Omega(n)}}$ fraction of $f$. 
\end{theorem}
\begin{proof}
To test equivalence of trilinear forms of $f, 
g:\F_q^n\times\F_q^n\times\F_q^n\to\F_q$, an average-case algorithm in time 
$q^{O(n)}$ can be achieved by following the proof of 
Theorem~\ref{thm:cubic_form}. The only difference is that, in Step 2 there, 
instead of symmetric matrices in Equation~\ref{eq:C_i}, we can construct general 
matrices $C_i=(\alpha_{i,j,k}')_{j, k\in [\ell]}$. Then we need a version of 
Lemma~\ref{lem:tech} for general matrices, which is already shown in 
\cite[Proposition 19 and 20]{LQ17}. It says that when $r=4$, a random $\vC\in 
\M(\ell, q)^4$ satisfies that $|\Aut(\vC)|\leq q^\ell$. We then proceed exactly as 
in Theorem~\ref{thm:cubic_form} for odd $q$, and for even $q$ we use the 
technique described in Section~\ref{app:main}. 
\end{proof}

\paragraph{Algebra isomorphism.} Let $V$ be a vector space. An algebra is a 
bilinear map $*:V\times V\to V$. 
This bilinear map $*$ is considered as the product. Algebras most studied are 
those with certain conditions on the product, including unital ($\exists v\in V$ 
such that $\forall u\in V$, $v* u=u$), associative ($(u* v)* w=u* 
(v* w)$), and commutative ($u* v=v* u$). The authors of \cite{AS05,AS06} 
study algebras satisfying these conditions. Here we consider 
algebras without such restrictions. Two algebras $*, \cdot: V\times V\to V$ 
are \emph{isomorphic}, if there exists $T\in\GL(V)$, such that $\forall u, v\in 
V$, 
$T(u)* T(v)=T(u\cdot v)$. 
As customary in computational algebra, an algebra is represented by its structure 
constants, i.e. suppose $V\cong \F^n$, and fix a basis $\{e_1, \dots, e_n\}$. Then 
$e_i* e_j=\sum_{k\in[n]}\alpha_{i,j,k}e_k$, and this $3$-way array 
$\tA=(\alpha_{i,j,k})$ records the structure constants of the algebra with 
product $*$. A random algebra over $\F_q$ denotes the setting when 
$\alpha_{i,j,k}$ are independently sampled from 
$\F_q$ uniformly at random.

\begin{theorem}\label{prop:ai}
Let $f:\F_q^n\times\F_q^n\to \F_q^n$ be 
a random algebra, and let $g:\F_q^n\times\F_q^n\to \F_q^n$ be an 
arbitrary algebra. There exists a deterministic algorithm that decides 
whether $f$ and $g$ are isomorphic in time $q^{O(n)}$, for all but at most 
$\frac{1}{q^{\Omega(n)}}$ fraction of $f$. 
\end{theorem}
\begin{proof}
Suppose we have two algebras $*, \cdot:\F_q^n\times\F_q^n\to\F_q^n$, 
represented by their structure constants. The proof strategy of 
Theorem~\ref{prop:trilinear} carries out to test algebra isomorphism in a 
straightforward fashion. The only difference is in the verification step (i.e. 
Line~\ref{enum:verify}). More specifically, we can write an algebra as an element 
in $(\F_q^n)^* \otimes \F_q^n\otimes \F_q^n$, where $(\F_q^n)^*$ is the 
dual space of $\F_q^n$.\footnote{Note that here we put $(\F_q^n)^*$ as the first 
argument, 
instead of the last one, in order to be consistent with the procedure in 
Proposition~\ref{prop:trilinear}. This is without loss of generality due to the 
standard isomorphism between $U\otimes V\otimes W$ and $W\otimes U\otimes V$.} It 
follows that we can write $*$ as 
$\sum_{i,j,k\in[n]}\alpha_{i,j,k}e_i^*\otimes e_j\otimes e_k$.
The key difference with trilinear form equivalence is that for \AI, $T\in \GL(n, 
q)$ 
acts on $e_i^*$ by its inverse. So the algorithm for \AI is the same as the one 
for trilinear form equivalence, except that in the verification step we need to 
use $R^{-1}$ instead of $R$ to act on the first argument. 
\end{proof}

\section{Complexity of symmetric and alternating trilinear form 
equivalence}\label{sec:complexity}

As mentioned in Section~\ref{subsec:tech}, the proof of 
Theorem~\ref{thm:complexity} follows by showing that symmetric and alternating 
trilinear form equivalence are \TI-hard (recall Definition~\ref{def:TI}). In the 
following we focus on the 
alternating case. The symmetric case can be tackled in a straightforward way, by 
starting from the \TI-complete problem, symmetric matrix tuple pseudo-isometry,  
from 
\cite[Theorem B]{GQ_arxiv}, and 
modifying the alternating gadget to a symmetric one.

\begin{proposition}\label{thm:alt_TI}
The alternating trilinear form equivalence problem is \TI-hard.
\end{proposition}

\begin{proof}

\paragraph{The starting \TI-complete problem.} We use the following \TI-complete 
problem from \cite{GQ_arxiv}. 
Let $\vA=(A_1, \dots, A_m), 
\vB=(B_1, \dots, B_m)\in \Lambda(n, \F)^m$ be two tuples of alternating matrices. 
We say that $\vA$ and $\vB$ are pseudo-isometric, if there exist $C\in\GL(n, \F)$ 
and $D=(d_{i,j})\in\GL(m, \F)$, such that for any $i\in[m]$, 
$\tr{C}(\sum_{j\in[m]}d_{i,j}A_j)C=B_i$. By \cite[Theorem B]{GQ_arxiv}, the 
alternating matrix tuple pseudo-isometry problem is \TI-complete. Without loss 
of generality, we assume that $\dim(\langle A_i\rangle)=\dim(\langle B_i\rangle)$, 
as if not, then they cannot be pseudo-isometric, and this dimension condition is 
easily checked.

An alternating trilinear form $\phi:\F^n\times \F^n\times \F^n\to \F$ naturally 
corresponds to a $3$-way array $\tA=(a_{i,j,k})\in \M(n\times n\times n, 
\F)$, where $a_{i,j,k}=\phi(e_i,e_j,e_k)$. Then $\tA$ is also alternating, i.e. 
$a_{i,j,k}=0$ if $i=j$ or $i=k$ or $j=k$, and 
$a_{i,j,k}=\sgn(\sigma)a_{\sigma(i), \sigma(j), \sigma(k)}$ for any $\sigma\in 
\S_3$. So in the following, we present a construction of an alternating 
$3$-way array from an 
alternating 
matrix tuple, 
in such a way that two alternating matrix tuples are pseudo-isometric if and only 
if the corresponding alternating trilinear forms are equivalent.

\paragraph{Constructing alternating $3$-way arrays from alternating matrix tuples.}
Given $\vA=(A_1, \dots, A_m)\in\Lambda(n, \F)^m$, 
we first build the $n \times n \times m$ tensor $\tA$ which has $A_1, \dotsc, A_m$ 
as its frontal slices.
Then we will use essentially the following construction twice in succession. 
We will give two viewpoints on this construction: one algebraic, in terms of 
trilinear forms, and another ``matricial'', in terms of 3-way arrays. Different 
readers may prefer one viewpoint over the other; our opinion is that the algebraic 
view makes it easier to verify the alternating property while the matricial view 
makes it easier to verify the reduction. We thank an anonymous review for the 
suggestion of the algebraic viewpoint.
The construction is, in some sense, the 3-tensor analogue of taking an ordinary 
matrix $A$ and building the alternating matrix $\begin{bmatrix} 0 & A \\ -\tr{A} & 
0 \end{bmatrix}$.

\emph{Notation:} Just as the transpose acts on matrices by $(\tr{A})_{i,j} = 
A_{j,i}$, for a 3-tensor $\tA$, we have six possible ``transposes'' corresponding 
to the six permutations of the three coordinates. Given $\sigma \in \S_3$, we 
write 
$\tA^\sigma$ for the 3-tensor defined by $(\tA^\sigma)_{i_1,i_2,i_3} = 
\tA_{i_{\sigma(1)}, i_{\sigma(2)}, i_{\sigma(3)}}$. 

Given a 3-way array $\tA \in \M(n \times m \times d, \F)$, we will make use of 
$\tA^{(23)}$ and $\tA^{(13)}$:
\begin{itemize}
\item $\tA^{(23)}$ is $n \times d \times m$ and has $\tA^{(23)}_{i,j,k} = 
\tA_{i,k,j}$. Equivalently, the $k$-th frontal slice of $\tA^{(23)}$ is the $k$-th 
vertical slice of $\tA$.

\item $\tA^{(13)}$ is $d \times m \times n$ and has $\tA^{(13)}_{i,j,k} = 
\tA_{k,j,i}$. Equivalently, the $k$-th frontal slice of $\tA^{(13)}$ is the 
transpose of the $k$-th horizontal slice of $\tA$.

%
%
%

\end{itemize}

\begin{example}[Running example]\label{ex:running}
Let us examine a simple example as follows. Let $\vA=(A)\in \Lambda(2, \F)^1$, 
where $A=\begin{bmatrix}
0 & a \\
-a & 0 
\end{bmatrix}$. Then $\tA=(A)$; $\tA^{(23)}=(A_1', A_2')\in \M(2\times 1\times 
2, \F)$, where 
$A_1'=\begin{bmatrix}
0 \\
-a
\end{bmatrix}$, and $A_2'=\begin{bmatrix}
a \\
0
\end{bmatrix}$; $\tA^{(13)}=(A_1'', A_2'')\in\M(1\times 2\times 2, \F)$, where 
$A_1''=\begin{bmatrix}
0 & a\\
\end{bmatrix}$, and $A_2''=\begin{bmatrix}
-a & 0 
\end{bmatrix}$.
\end{example}

From the above $\tA$, $\tA^{(23)}$, and $\tA^{(13)}$, we construct $\tilde\tA\in 
\M((n+m)\times (n+m)\times (n+m), \F)$ as follows. 
We divide $\tilde\tA$ 
into 
the following eight blocks. That is, set $\tilde\tA=(\tilde\tA_1, \tilde\tA_2)$ 
(two block frontal slices) 
where 
 $\tilde\tA_1=\begin{bmatrix}
0_{n \times n \times n} & \tA^{(23)} \\
\tA^{(13)} & 0 
\end{bmatrix}$, and $\tilde\tA_2=\begin{bmatrix}
-\tA & 0 \\
0 & 0_{m \times m \times m}
\end{bmatrix}$, where $0_{n \times n \times n}$ indicates the $n \times n \times 
n$ zero tensor, and analogously for $0_{m \times m \times m}$ (the remaining sizes 
can be determined from these and the fact that $\tA$ is $n \times n \times m$).

The corresponding construction on trilinear forms is as follows. The original 
trilinear form is $A(x,y,z)  = \sum_{i,j \in [n],k \in [m]} a_{i,j,k} x_i y_j 
z_k$, where $x = (x_1, \dotsc, x_n)$, $y = (y_1, \dotsc, y_n)$, and $z=(z_1, 
\dotsc, z_m)$, and we have $A(x,y,z) = -A(y,x,z)$. The new trilinear form will be 
$\tilde{A}(x', y', z')$, where 
\begin{eqnarray*}
x'=(x^{(1)}, x^{(2)}) & = & (x^{(1)}_1, \dotsc, x^{(1)}_n, x^{(2)}_1, \dotsc, 
x^{(2)}_m) \\
y' = (y^{(1)}, y^{(2)}) & = & (y^{(1)}_1, \dotsc, y^{(1)}_n, y^{(2)}_1, \dotsc, 
y^{(2)}_m) \\
z' = (z^{(1)}, z^{(2)}) & = & (z^{(1)}_1, \dotsc, z^{(1)}_n, z^{(2)}_1, \dotsc, 
z^{(2)}_m). 
\end{eqnarray*}
This new form will satisfy $\tilde{A}(x', y', z') = \sum_{i,j,k \in [n+m]} 
\tilde{a}_{i,j,k} x'_i y'_j z'_k$. Let us unravel what this looks like from the 
above description of $\tilde\tA$. We have
\begin{eqnarray*}
\tilde{A}(x',y',z') & = & \sum_{i \in [n], j \in [m], k \in [n]} 
(\tilde\tA_1)_{i,n+j,k} x'_i y'_{n+j} z'_k + \sum_{i \in [m], j,k \in [n]} 
(\tilde\tA_1)_{n+i, j, k} x'_{n+i} y'_j z'_k \\
&&\quad + \sum_{i,j \in [n], k \in [m]} (\tilde\tA_2)_{i,j,k} x'_i y'_j z'_{n+k} \\
& = & \sum_{i \in [n], j \in [m], k \in [n]} \tA^{(23)}_{i,j,k} x'_i y'_{n+j} z'_k 
+ \sum_{i \in [m], j,k \in [n]} \tA^{(13)}_{i, j, k} x'_{n+i} y'_j z'_k - 
\sum_{i,j \in [n], k \in [m]} \tA_{i,j,k} x'_i y'_j z'_{n+k} \\
 & = & \sum_{i \in [n], j \in [m], k \in [n]} \tA_{i,k,j} x'_i y'_{n+j} z'_k + 
 \sum_{i \in [m], j,k \in [n]} \tA_{k, j, i} x'_{n+i} y'_j z'_k - \sum_{i,j \in 
 [n], k \in [m]} \tA_{i,j,k} x'_i y'_j z'_{n+k} \\
  & = & A(x^{(1)}, z^{(1)}, y^{(2)})  + A(z^{(1)}, y^{(1)}, x^{(2)}) - A(x^{(1)}, 
  y^{(1)}, z^{(2)})
\end{eqnarray*}
From this formula, and the fact that $A(x,y,z) = -A(y,x,z)$, we can now more 
easily verify that $\tilde{A}$ is alternating in all three arguments. Since the 
permutations $(13)$ and $(23)$ generate $S_3$, it suffices to verify it for these 
two. We have
\begin{eqnarray*}
\tilde{A}^{(13)}(x',y',z') & = & \tilde{A}(z', y', x') \\
& = & A(z^{(1)}, x^{(1)}, y^{(2)})  + A(x^{(1)}, y^{(1)}, z^{(2)}) - A(z^{(1)}, 
y^{(1)}, x^{(2)}) \\
& = & -A(x^{(1)}, z^{(1)}, y^{(2)}) + A(x^{(1)}, y^{(1)}, z^{(2)}) - A(z^{(1)}, 
y^{(1)}, x^{(2)}) \\
& = & -\tilde{A}(x',y',z').
\end{eqnarray*}
Similarly, we have:
\begin{eqnarray*}
\tilde{A}^{(23)}(x',y',z') & = & \tilde{A}(x', z', y') \\
 & = & A(x^{(1)}, y^{(1)}, z^{(2)})  + A(y^{(1)}, z^{(1)}, x^{(2)}) - A(x^{(1)}, 
 z^{(1)}, y^{(2)}) \\
 & = & A(x^{(1)}, y^{(1)}, z^{(2)})  - A(z^{(1)}, y^{(1)}, x^{(2)}) - A(x^{(1)}, 
 z^{(1)}, y^{(2)}) \\
 & = & -\tilde{A}(x',y',z'),
\end{eqnarray*}
as claimed.

%


\begin{example}[Running example, continued from Example~\ref{ex:running}] 
\label{ex:running2}
We can write out $\tilde\tA$ in this case explicitly. The first block frontal 
slice $\tilde\tA_1$ is $3 \times 3 \times 2$, consisting of the two frontal slices
\[
\left(\begin{array}{cc;{2pt/2pt}c}
0 & 0 & 0\\
0 & 0 & -a\\ \hdashline[2pt/2pt]
0 & a & 0 
\end{array}\right)
\text{ and } 
\left(\begin{array}{cc;{2pt/2pt}c}
0 & 0 & a\\
0 & 0 & 0 \\ \hdashline[2pt/2pt]
-a & 0 & 0 
\end{array}\right)
\]
while the second block frontal slice $\tilde\tA_2$ is the $3 \times 3 \times 1$ 
matrix
\[
\left(\begin{array}{cc;{2pt/2pt}c}
0 & -a & 0 \\
a & 0 & 0 \\ \hdashline[2pt/2pt]
0 & 0 & 0 
\end{array}\right)
\]
It can be verified easily that $\tilde\tA=(a_{i,j,k})$ is 
alternating: the nonzero entries are $a_{2,3,1}=-a$, $a_{3,2,1}=a$, $a_{1,3,2}=a$, 
$a_{3,1,2}=-a$, $a_{1,2,3}=-a$, and $a_{2,1,3}=a$, which are consistent with the 
signs of the permutations.
\end{example}

\paragraph{The gadget construction.} We now describe the gadget construction. The 
gadget can be described as a block 
$3$-way array as follows. 
Construct a $3$-way array $\tG$ of size 
$(n+1)^2\times (n+1)^2\times (n+m)$ over $\F$ as follows. For $i\in[n]$, 
the $i$th frontal slice of $\tG$ is 
$$\begin{bmatrix}
0 & 0 & \dots & 0 & I_{n+1} & 0 & \dots & 0 \\
0 & 0 &  \dots & 0 & 0 & 0 & \dots & 0 \\
\vdots & \vdots &  \dots & \vdots & \vdots & \vdots & \dots & \vdots \\
0 & 0 &  \dots & 0 & 0 & 0 & \dots & 0 \\
-I_{n+1} & 0 &  \dots & 0 & 0 & 0 & \dots & 0 \\
0 & 0 &  \dots & 0 & 0 & 0 & \dots & 0 \\
\vdots& \vdots & \dots & \vdots & \vdots & \vdots & \dots & \vdots \\
0 & 0 &  \dots & 0 & 0 & 0 & \dots & 0 
\end{bmatrix},$$
where $0$ here denotes the $(n+1)\times (n+1)$ all-zero 
matrix, $I_{n+1}$ is at the $(1, i+1)$th block position, and $-I_{n+1}$ is at the 
$(i+1, 1)$th block position. 
For $n+1\leq i\leq n+m$, the $i$th frontal slice of $\tG$ is the all-zero matrix. 
We also need the following $3$-way arrays derived from $\tG$. 
We will use 
$\tG^{(13)}$ 
and $\tG^{(23)}$. 
Note that $\tG^{(13)}$ is
of size $(n+m)\times (n+1)^2\times 
(n+1)^2$, and its $i$th horizontal slice is the $i$th frontal 
slice of $\tG$. Similarly, $\tG^{(23)}$ is of size $ 
(n+1)^2\times (n+m)\times (n+1)^2$, 
and its $j$th vertical slice is the $j$th frontal slice of $\tG$.

Finally, construct a $3$-tensor $\hat\tA$ as follows. 
It consists of the two block frontal slices
\[
\begin{bmatrix}
\tilde \tA & 0 \\
0 & -\tG
\end{bmatrix}
\text{ and }
\begin{bmatrix}
0 & \tG^{(13)} \\
\tG^{(23)} & 0\\
\end{bmatrix}.
\]

To see how this all fits together, let $\tG_1$ be the $(n+1)^2 \times (n+1)^2 
\times n$ tensor consisting of the first $n$ frontal slices of $\tG$ (these are 
the only nonzero frontal slices of $\tG$). Then we may view $\hat\tA$ as having 
three block frontal slices, namely:
\[
\begin{bmatrix}
0_{n \times n \times n} & \tA^{(23)} & 0\\
\tA^{(13)} & 0_{m \times m \times n} & 0 \\
0 & 0 & -\tG_1
\end{bmatrix}, 
\begin{bmatrix}
-\tA & 0 & 0\\
0 & 0_{m \times m \times m} & 0 \\
0 & 0 & 0_{(n+1)^2 \times (n+1)^2 \times m}
\end{bmatrix},
\]
and 
\[
\begin{bmatrix}
0_{n \times n \times (n+1)^2} & 0 & \tG_1^{(13)} \\
0 & 0_{m \times m \times (n+1)^2} & 0 \\
\tG_1^{(23)} & 0 & 0\\
\end{bmatrix}.
\]

We claim that $\hat\tA$ is alternating. To verify this is straightforward but 
somewhat tedious. So we use the following example from which a complete proof can 
be extracted easily. 

\begin{example}[Running example, continued from Example~\ref{ex:running2}]
Let $\tA$ be the $2 \times 2 \times 1$ tensor with alternating frontal slice $A = 
\begin{bmatrix} 0 & a \\ -a & 0 \end{bmatrix}$. In particular, $n=2, m=1$, so 
$\tG$ will have size $(n+1)^2 \times (n+1)^2 \times (n+m) = 9 \times 9 \times 3$, 
and $\tA$ will have size $n+m+(n+1)^2 = 12$ in all three directions. We will write 
out the first $n+m=3$ frontal slices explicitly, as those are the only ones 
involving $\tA$, and leave the last 9 (involving only transposes of $\tG_1$) 
unwritten.
\[
\left(\begin{array}{cc;{2pt/2pt}c;{2pt/2pt}ccc}
0 & 0 & 0\\
0 & 0 & -a\\ \hdashline[2pt/2pt]
0 & a & 0 \\ \hdashline[2pt/2pt]
 &  &  & 0_3 & I_3 & 0 \\
   & & & -I_3 & 0_3 & 0 \\
   & & & 0 & 0 & 0_3
\end{array}\right), 
\left(\begin{array}{cc;{2pt/2pt}c;{2pt/2pt}ccc}
0 & 0 & a\\
0 & 0 & 0\\ \hdashline[2pt/2pt]
-a & 0 & 0 \\ \hdashline[2pt/2pt]
 &  &  & 0_3 & 0 & I_3 \\
   & & & 0 & 0_3 & 0 \\
   & & & -I_3 & 0 & 0_3
\end{array}\right),
\]
\[
\text{ and } \left(\begin{array}{cc;{2pt/2pt}c;{2pt/2pt}ccc}
0 & a & 0\\
-a & 0 & 0\\ \hdashline[2pt/2pt]
0 & 0 & 0 \\ \hdashline[2pt/2pt]
 &  &  & 0_3 & 0 & 0 \\
   & & & 0 & 0_3 & 0 \\
   & & & 0 & 0 & 0_3
\end{array}\right)
\]
and the remaining 9 frontal slices look like
\[
\left(\begin{array}{cc;{2pt/2pt}c;{2pt/2pt}ccc}
0 & 0 & 0 & \\
0 & 0 & 0 & \multicolumn{3}{c}{\tG_1^{(13)}}\\ \hdashline[2pt/2pt]
0 & 0 & 0 & \multicolumn{3}{c}{0_{1 \times 9 \times 9}} \\ \hdashline[2pt/2pt]
 &  &  & 0_{3 \times 3 \times 9} & 0 & 0 \\
 \multicolumn{2}{c}{\tG_1^{(23)}} & 0_{9 \times 1 \times 9} & 0 & 0_{3 \times 3 
 \times 9} & 0 \\
   & & & 0 & 0 & 0_{3 \times 3 \times9}
\end{array}\right)
\]
Since the $a$'s only appear in positions with the same indices as they did in 
$\tilde\tA$ (see Example~\ref{ex:running2}), that portion is still alternating. 
For the $\tG$ parts, note that the identity matrices in the first three frontal 
slices, when having their indices transposed, end up either in the $\tG_1^{(13)}$ 
portion or the $\tG_1^{(23)}$ portion, with appropriate signs.
\end{example}


\paragraph{Proof of correctness.} Let $\vA, \vB\in\Lambda(n, \F)^m$. Let 
$\hat\tA=(\begin{bmatrix}
\tilde\tA & 0 \\
0 & -\tG
\end{bmatrix},\begin{bmatrix}
0 & \tG^{(13)} \\
\tG^{(23)}& 0 
\end{bmatrix})$, $
\hat\tB=(\begin{bmatrix}
\tilde\tB & 0 \\
0 & -\tG
\end{bmatrix},\begin{bmatrix}
0 & \tG^{(13)} \\
\tG^{(23)}& 0 
\end{bmatrix})\in\M((n+m+(n+1)^2)\times (n+m+(n+1)^2)\times 
(n+m+(n+1)^2), \F)$ be constructed from $\vA$ and $\vB$ using the procedure 
above, respectively. 

We claim that $\vA$ and $\vB$ are pseudo-isometric if and only if $\hat\tA$ and 
$\hat\tB$ are equivalent as trilinear forms. 

\paragraph{The only if direction.} Suppose $\tr{P}\vA P=\vB^Q$ for some 
$P\in\GL(n, 
\F)$ and $Q\in\GL(m, \F)$. 

We will construct a trilinear form equivalence from $\hat\tA$ to $\hat\tB$ of the 
form  
$S=\begin{bmatrix}
P & 0 & 0 \\
0 & Q^{-1} & 0 \\
0 & 0 & R
\end{bmatrix}\in \GL(n+m+(n+1)^2, \F)$, where $R\in\GL((n+1)^2, \F)$ is to be 
determined later on. 

Recall that $\hat\tA=(\begin{bmatrix}
\tilde\tA & 0 \\
0 & -\tG
\end{bmatrix},\begin{bmatrix}
0 & \tG^{(13)} \\
\tG^{(23)}& 0 
\end{bmatrix})$, $
\hat\tB=(\begin{bmatrix}
\tilde\tB & 0 \\
0 & -\tG
\end{bmatrix},\begin{bmatrix}
0 & \tG^{(13)} \\
\tG^{(23)}& 0 
\end{bmatrix})$. It can be verified that the action of $S$ sends $\tilde\tA$ to 
$\tilde\tB$. It 
remains to show that, by choosing an appropriate $R$, the action of $S$ also sends 
$\tG$ to $\tG$.

Let $\tG_1$ 
be the first $n$ frontal slices of $\tG$, and $\tG_2$ 
the last $m$ frontal slices from $\tG$. Then the action of $S$ sends $\tG_1$ to 
$\tr{R}\tG_1^P R$, and $\tG_2$ to $\tr{R}\tG_2^{Q^{-1}}R$.
Since $\tG_2$ is all-zero, the action of $S$ on $\tG_2$ results in an all-zero 
tensor, so we have $\tr{R}\tG_2^{Q^{-1}}R=\tG_2$. 

We then turn to $\tG_1$. For $i\in[n+1]$, consider the $i$th horizontal slice of 
$\tG_1$, which is 
of the form $H_i=\begin{bmatrix}
0 & B_{1, i} & B_{2, i} & \dots & B_{n, i}
\end{bmatrix}$, where $0$ denotes the $n\times (n+1)$ all-zero matrix, and 
$B_{j, i}$ is the $n\times (n+1)$ elementary matrix with the 
$(j, i)$th
entry being $1$, and other entries being $0$. Note that those non-zero entries of 
$H_i$ are 
in the $(k(n+1)+i)$th columns, for $k\in[n]$. Let $\tr{P}=\begin{bmatrix}
p_1 & \dots & p_n
\end{bmatrix}$, where $p_i$ is the $i$th column of $\tr{P}$. Then $P$ acts on 
$H_i$ from the 
left, which yields $\tr{P}H_i=\begin{bmatrix}
0 & P_{1, i} & \dots & P_{n, i}
\end{bmatrix}$, where $P_{j, i}$ denotes the $n\times (n+1)$ matrix with the $i$th 
column being 
$p_j$, and the other columns being $0$. 

Let us first set $R=\begin{bmatrix}
I_{n+1} & 0 \\
0 & \hat{R}
\end{bmatrix}$, where $\hat{R}$ is to be determined later on. Then the left action 
of $R$ on $\tG_1$ preserves $H_i$ through 
$I_{n+1}$. The right action of $R$ on $\tG_1$ translates to the right action of 
$\hat{R}$ on $H_i$. To send $\tr{P}H_i$ back to $H_i$, $\hat{R}$ needs to
act on 
those $(k(n+1)+i)$th columns of $H_i$, $i\in[n+1]$, as $P^{-1}$. Note that for 
$H_i$ 
and 
$H_j$, $i\neq j$, those columns with non-zero entries are disjoint. This gives 
$\hat{R}$ the freedom to handle different $H_i$'s separately. In other words, 
$\hat{R}$ can be set as $P^{-1}\otimes I_{n+1}$. This ensures that for every 
$H_i$, 
$\tr{P}H_i\hat{R}=H_i$. To summarize, we have $\tr{R}\tG_1^P R=\tG_1$, and this 
concludes the proof for the only if direction. 

\paragraph{The if direction.} Suppose $\hat\tA$ and $\hat\tB$ are isomorphic as 
trilinear forms via $P\in\GL(n+m+(n+1)^2, \F)$. Set $P=\begin{bmatrix}
P_{1,1} & P_{1,2} & P_{1,3}\\
P_{2,1} & P_{2,2} & P_{2,3}\\
P_{3,1} & P_{3,2} & P_{3,3}
\end{bmatrix}$, where $P_{1,1}$ is of size $n\times n$, $P_{2,2}$ is of size 
$m\times m$, and $P_{3,3}$ is of size $(n+1)^2\times (n+1)^2$. 
Consider the ranks of the frontal slices of $\hat{\tA}$.
\begin{itemize}
\item The ranks of the first $n$ frontal slices are in $[2(n+1), 4n]$. This is 
because a frontal slice in this range consists of two copies of vertical slices 
of $\tA$ (whose ranks are between $[0, n-1]$ due to the alternating condition), 
and one frontal slice of $\tG$ (whose ranks are of $2(n+1)$).
\item The ranks of the $n+1$ to $n+m$ frontal slices are in $[0, 
n]$. This is because a frontal slice in this range consists of only just one 
frontal slice of $\tA$.
\item The ranks of the last $n(n+1)$ vertical slices are in $[0, 2n]$. This is 
because a frontal slice in this range consists of two copies of horizontal slices 
of $\tG$ (whose ranks are either $n$ or $1$; see e.g. the form of $H_i$ in the 
proof of the only if direction).
\end{itemize}
By the discussions above, we claim that that $P$ must be of the form 
$\begin{bmatrix}
P_{1,1} & 0 & 0\\
P_{2,1} & P_{2,2} & P_{2,3}\\
P_{3,1} & P_{3,2} & P_{3,3}
\end{bmatrix}$. To see this, for the sake of contradiction, suppose there are 
non-zero entries in $P_{1,2}$ or 
$P_{1,3}$. Then a non-trivial linear combination of the first $n$ frontal slices 
is added to one of the last $(m+(n+1)^2)$ frontal slices. This implies that for 
this slice, the lower-right $(n+1)^2\times (n+1)^2$ submatrix is of the form 
$\begin{bmatrix}
0 & a_1 I_{n+1} & a_2 I_{n+1} & \dots & a_nI_{n+1} \\
-a_1 I_{n+1} & 0 &  0 &  \dots& 0  \\
-a_2 I_{n+1} & 0 & 0 &  \dots  & 0 \\
\vdots & \vdots & \vdots & \ddots  & \vdots\\
-a_n I_{n+1} & 0& 0 &  \dots  & 0 
\end{bmatrix}$, where one of $a_i\in \F$ is non-zero. Then this slice is of rank 
$\geq 2(n+1)$, which is unchanged by left (resp. right) multiplying $\tr{P}$ 
(resp. $P$), so it cannot be equal to the corresponding slice of $\hat\tB$ 
which is of rank $\leq 2n$. We then arrived at the desired contradiction.

Now consider the action of such $P$ on the $n+1$ to $n+m$ frontal slices. Note 
that these slices are of the form $\begin{bmatrix}
A_i & 0 & 0 \\
0 & 0 & 0 \\
0 & 0 & 0
\end{bmatrix}$. (Recall that the last $m$ slices of $\tG$ are all-zero matrices.) 
Then we have 
$
\begin{bmatrix}
\tr{P_{1,1}} & \tr{P_{2,1}} & \tr{P_{3,1}}\\
0 & \tr{P_{2,2}} & \tr{P_{3,2}}\\
0 & \tr{P_{2,3}} & \tr{P_{3,3}}
\end{bmatrix}
\begin{bmatrix}
A_i & 0 & 0 \\
0 & 0 & 0 \\
0 & 0 & 0
\end{bmatrix}
\begin{bmatrix}
P_{1,1} & 0 & 0\\
P_{2,1} & P_{2,2} & P_{2,3}\\
P_{3,1} & P_{3,2} & P_{3,3}
\end{bmatrix}=
\begin{bmatrix}
\tr{P_{1,1}}A_iP_{1,1} & 0 & 0 \\
0 & 0 & 0 \\
0 & 0 & 0
\end{bmatrix}.$
Since $\tr{P}\hat\tA^P P=\hat\tB$, we have $\tr{P}\hat\tA P=\hat\tB^{P^{-1}}$. 
Observe that for the upper-left $n\times n$ submatrices of the frontal slices of 
$\hat\tB$, $P^{-1}$ simply performs a linear combination of $B_i$'s. 
It 
follows that every $\tr{P_{1,1}}A_iP_{1,1}$ is in the linear span of $B_i$. Since 
we assumed 
$\dim(\langle 
A_i\rangle)=\dim(\langle B_i\rangle)$, we have that $\vA$ 
and $\vB$ are pseudo-isometric. 
This concludes the proof of Proposition~\ref{thm:alt_TI}.
\end{proof}

\section*{Acknowledgment}
  \noindent The authors thank the anonymous reviewers for their careful reading 
  and  helpful suggestions. 


\end{document}